\documentclass[10pt,onecolumn]{IEEEtran}
\usepackage{amsfonts}
\usepackage{indentfirst}
\usepackage{amsmath}
\usepackage{amsthm}
\usepackage{amssymb}
\usepackage{color}
\usepackage{graphicx}
\usepackage[para]{threeparttable}
\usepackage{booktabs}
\usepackage{caption}
\usepackage{makecell}
\usepackage{diagbox}
\usepackage{multirow}
\usepackage{epstopdf}
\usepackage{subfigure}
\usepackage{threeparttable}
\usepackage{algorithm}
\usepackage{algpseudocode}
\usepackage{cite}

\theoremstyle{remark} 

\newtheorem{Corollary}{\textit{Corollary}}
\newtheorem{Definition}{\textit{Definition}}
\newtheorem{Theorem}{\textit{Theorem}}
\newtheorem{Lemma}{\textit{Lemma}}
\newtheorem{Remark}{\textit{Remark}}
\newtheorem{Example}{\textit{Example}}

\newtheorem{Construction}{\textit{Construction}}
\def\mathbf{\boldsymbol}

\allowdisplaybreaks[4]
\begin{document}
\title{Some Optimal and Near Optimal Doppler Resilient Complementary Sequence Sets}
\author{
Bingsheng Shen,~\IEEEmembership{Member,~IEEE},
Yang Yang,~\IEEEmembership{Member,~IEEE},
Zhengchun Zhou,~\IEEEmembership{Member,~IEEE},
Pingzhi Fan,~\IEEEmembership{Fellow,~IEEE}
\thanks{B. Shen, Z. Zhou and P. Fan are with the School of Information Science and Technology, Southwest Jiaotong University, Chengdu, 611756, China. E-mail: bsshen9527@swjtu.edu.cn, zzc@swjtu.edu.cn, pzfan@swjtu.edu.cn; Y. Yang is with the School of Mathematics, Southwest Jiaotong University, Chengdu, 611756, China. E-mail: yang\_data@swjtu.edu.cn.}
}

\markboth{Journal of \LaTeX\ Class Files,~Vol.~14, No.~8, August~2021}%
{Shell \MakeLowercase{\textit{et al.}}: A Sample Article Using IEEEtran.cls for IEEE Journals}

\maketitle
\begin{abstract}
 Sequences with excellent ambiguity functions are very useful in radar detection and modern mobile communications. Doppler resilient complementary sequence (DRCS) is a new type of sequence proposed recently, which can achieve lower ambiguity function sidelobes by summing the ambiguity functions of subsequences. In this paper, we introduce some new constructions of DRCS sets (DRCSSs) based on one-coincidence frequency-hopping sequence sets (OC-FHSSs), almost difference sets (ADSs), some specific sequences, etc. Critically, the proposed DRCSSs are optimal or near optimal.
\end{abstract}


\section{Introduction}
Sequences with low correlation is widely used in engineering, especially in wireless communication and radar sensing. 
Such as synchronization, channel estimation, interference cancellation, sensing, positioning, etc \cite{golomb2005signal,fan1996sequence}. 
The ideal sequence set is one in which the auto-correlation function of each sequence is an impulse function, and the cross-correlation function of any different sequences is always equal to zero. 
However, according to the Welch bound, there is no such sequence set \cite{welch1974lower}. 

In order to achieve ideal correlation functions, a concept called ``complementary sequence (CS)'' has been developed and widely studied. 
The idea of ``complementarity'' was first introduced by Marcel J. E. Golay in the 1960s which aims to achieve an ideal aperiodic auto-correlation function by summing the auto-correlation functions of two binary sequences \cite{golay1961comp}. A large number of researchers have studied the properties and construction of binary/quaternary CSs \cite{turyn1974had,griffin1977there,eliahou1990new,kounias1991golay,craigen2002com,
borwein2004com,fiedler2010new,gibson2011qua1,gibson2011qua2}.
Following Golay's research, Tseng and Liu proposed the concept of mutually orthogonal CS sets (MOCSSs), where each CS contains multiple subsequences and the sums of correlation functions of all subsequences can achieve ideal auto- and cross-correlation properties \cite{tseng1972compl}. 
In a MOCSS, the set size $K$ is limited by the number $M$ (is also celled flock size) of subsequences in each CS, that is, $K\le M$. In particular, a MOCSS is called a complete complementary code (CCC) \cite{suehiro1988n} if $K=M$. 
In order to increase the set size, Liu \emph{et al.} proposed the concept of low-correlation-zone CSS (LCZ-CSS), which only considers the amplitude of the correlation function of CSS in a LCZ \cite{liu2011corr}. These sequence sets have received a lot of attention in engineering, and there are many mathematical construction methods.

Modern sequence design is more stringent as one is expected to deal with the notorious Doppler effect in various  mobile channels \cite{wu2016survey,duggal2020dop}. 
For example, in connected vehicles and radar sensing, the received signal is often a time-shifted and phase-rotated version of the transmitted signal due to propagation delay and mobility-incurred Doppler. 
In signal processing, the ambiguity function (AF) is usually used to characterize the delay-Doppler response of the receiver \cite{wood1953pro}. 
Precise delay-Doppler estimation requires sequence sets that simultaneously minimize both the peak sidelobe level of the auto-AF and the peak magnitude of the cross-AF. 
In \cite{ding2013unit}, Ding \emph{et al.} first gave a theoretical bound on the AF amplitude of a sequence set. However, this bound is not tight and there is no sequence set that can reach the bound. 
Later, Ye \emph{et al.} improved the bound in \cite{ding2013unit} by utilizing the property that the auto-AF at zero delay in a unimodular sequence is equal to zero for any non-zero Doppler \cite{ye2022low}. 
Besides, Ye \emph{et al.} also referred to such the sequence set as the Doppler-resilient sequence set (DRSS). If the parameters of a DRSS can make the equality of the theoretical lower bound hold, it is called optimal. Designing optimal DRSS is a challenge. 
Some optimal DRSSs or DRSSs with low AF amplitudes can be found in \cite{ding2013unit,ye2022low,wang2011new,wang2013new,tian2025asym,wang2025new}. 

To minimize AF's sidelobes, we drew inspiration from DRSSs and complementary sequences and proposed the concept of a Doppler-resilient complementary sequence set (DRCSS) \cite{shen2025drcs}.
The AF of each Doppler resilient complementary sequence (DRCS) is characterized by the sum of the AFs of its subsequences. 
The idea of DRCS originated from \cite{turyn1963ambiguity}, Turyn studied the mean square value of the AF of binary DRCS containing two subsequences. 
Compared with DRSSs, DRCSSs can achieve lower AF sidelobes, even thumbtack-type AF. The research on DRCSSs is still in its initial stage, and there are relatively few related achievements. In this paper, we propose several constructions that can generate optimal or almost optimal periodic DRCSSs. 
Firstly, based on the observation of circular Florentine rectangle (CFR), we found that it is essentially an one-coincidence frequency-hopping sequence set (OC-FHSS). Therefore, we directly use the OC-FHSS to obtain many asymptotically optimal DRCSSs with new parameters. 
Subsequently, we adapt the framework established in \cite{shen2025drcs} to generate asymptotically optimal or near-optimal DRCSSs by using ADSs.
Finally, we find that by repeating the complete complementary codes (CCCs), the optimal DRCSSs with zero ambiguity zone (ZAZ) can be obtained. The advantage of this construction is that it can generate DRCSSs for small alphabets.
As a comparison with the known constructions, the parameters of our proposed periodic DRCSSs are listed in Table \ref{tab-knowndrcss}. It is shown that our proposed DRCSSs are asymptotically optimal or near optimal with respect to the theoretical bounds in \cite{shen2025drcs}.

\begin{table}[!t]
\centering
\caption{Known parameters of periodic DRCSSs}\label{tab-knowndrcss}
\resizebox{.99\columnwidth}{!}{
\begin{tabular}{|c|c|c|c|c|c|c|c|c|}
\hline
Source & Set size & Flock size & Length & $\theta_{\max}$ & $Z_x$ & $Z_y$ & Constraints & Optimality  \\ \hline
\cite{shen2025drcs} & $1$ & $N$ & $N$ & $0$ & $N$ & $N$ & $N\ge2$ is an integer & optimal \\ \hline
\cite{shen2025drcs} & $P$ & $PN$ & $N$ & $0$ & $N$ & $N$ & \makecell{$P\ge2$ and $N\ge2$ are \\ two integers} & optimal  \\ \hline
\cite{shen2025drcs} & $p_0-1$ & $N$ & $N$ & $N$ & $N$ & $N$ & \makecell{$N\ge2$ is an integer and $p_0$ is \\the smallest prime factor of $N$} & optimal  \\ \hline
\cite{shen2025drcs} & $K$ & $M$ & $N$ & $\sqrt{\frac{NM(N-M)}{N-1}}$ & $N$ & $\lfloor N/K \rfloor$ & \makecell{$1\le K\le N$ is an integer,\\ $N=2^n-1$ is a prime and \\ there is a $(N,M,\lambda)$-DS} & optimal  \\ \hline
Corollary \ref{cor-2} & $2^m-2$ & $2^m-1$ & $2^m-1$ & $2^m-1$ & $2^m-1$ & $2^m-1$ & $m\ge2$ is an integer & optimal  \\ \hline
Corollary \ref{cor-2} & $p^m$ & $p^m$ & $p^m-1$ & $p^m$ & $p^m$ & $p^m$ & \makecell{$p$ is a prime and \\ $m\ge2$ is an integer} & optimal  \\ \hline
Corollary \ref{cor-2} & $p$ & $p^2$ & $p(p-1)$ & $p^2$ & $p(p-1)$ & $p(p-1)$ & $p$ is a prime & optimal  \\ \hline
Theorem \ref{the-cfr} & $p_0-1$ & $N-1$ & $N$ & $N$ & $N$ & $N$ & \makecell{$N\ge2$ is an integer and $p_0$ is \\the smallest prime factor of $N$} & optimal  \\ \hline
Theorem \ref{the-drcss-ads} & $K$ & $\frac{p-1}{2}$ & $p$ & $\frac{p+\sqrt{p}}{2}$ & $p$ & $\lfloor p/K \rfloor$ & \makecell{$1\le K\le p$ is an integer and\\ $p=1\pmod 4$ is a prime} & optimal  \\ \hline
Corollary \ref{cor-ads} & $1$ & $\frac{p-1}{2}$ & $p-1$ & $\frac{p-1}{\sqrt{2}}$ & $p-1$ & $p-1$ & \makecell{ $p=3\pmod 4$ is a prime} & near-optimal  \\ \hline
Corollary \ref{cor-ads} & $1$ & $\frac{p-1}{2}$ & $p-1$ & $p-1$ & $p-1$ & $p-1$ & $p=1\pmod 4$ is a prime & near-optimal  \\ \hline
Theorem \ref{the-drcssccc} & $M$ & $M$ & $LN$ & $0$ & $N$ & $L$ & \makecell{there is a $(M,N)$-CCC and\\ $L\ge1$ is an integer} & optimal  \\ \hline
\end{tabular}}
\end{table}

The structure of this paper is outlined as follows: In Section II, we introduce the essential mathematical tools and notations utilized in this work. In Section III, we present the main results of this paper and provide some examples. Lastly, Section IV concludes this paper.

\emph{Notations:} $\xi_N=\exp(2\pi\sqrt{-1}/N)$ denotes the primitive $N$-th complex root of unit; $\mathbb{Z}_N=\{0,1,\cdots,N-1\}$ is a ring of integers modulo $N$; $\otimes$ denotes the Kronecker product; $x^{*}$ denotes the complex conjugate of $x$;

\section{Preliminaries}
In this section, we will introduce the definition of DRCSSs and review the corresponding theoretical bounds. In addition, we will introduce some mathematical tools used in this paper. 

\subsection{Ambiguity Functions}
A conventional sequence set $\mathbf{S}$ contains $K$ sequences of length $N$, and each of which has all of its entries taking values from a set, called the alphabet. Namely, $\mathbf{S}=\{\mathbf{s}_k:0\le k\le K-1\}$, $\mathbf{s}_k=(s_{k,0},s_{k,1},\cdots,s_{k,N-1})$, where $0\le k<K$. In this paper, we consider unimodular sequences, that is, each element in the sequence is defined as a complex root of unity. For two sequence $\mathbf{s}_u$ and $\mathbf{s}_v$ in $\mathbf{S}$, their periodic cross ambiguity function (PCAF) is defined as 
\begin{align}
  AF_{\mathbf{s}_u,\mathbf{s}_v}(\tau,f) = \sum\limits_{i=0}^{N-1}s_{u,i}s_{v,i+\tau}^{*}\xi_N^{fi},
\end{align}
where $-N<\tau,f< N$, the summation $i+\tau$ is calculated on $\mathbb{Z}_N$. If $u=v$, $AF_{\mathbf{s}_u,\mathbf{s}_v}(\tau,f)$ is called the periodic auto ambiguity function (PAAF), and denoted by $AF_{\mathbf{s}_u}(\tau,f)$. 
Especially, when $f=0$, the two-dimensional AF degenerates into a one-dimensional correlation function, with auto- and cross-correlation function denoted as $R_{\mathbf{s}_u}(\tau)$ and $R_{\mathbf{s}_u,\mathbf{s}_v}(\tau)$, respectively.

\subsection{Doppler Resilient Complementary Sequence Sets}
A DRCSS $\mathcal{S}$ contains $K$ (i.e., set size) DRCSs, each DRCS consists of $M$ (i.e., flock size) subsequences of length $N$, i.e., $\mathcal{S}=\{\mathbf{S}^{(k)}:0\le k<K\}$, $\mathbf{S}^{(k)}=\{\mathbf{s}_m^{(k)}:0\le m<M\}$, $\mathbf{s}_m^{(k)}=(s_{m,0}^{(k)},s_{m,1}^{(k)},\cdots,s_{m,N-1}^{(k)})$. Sometimes, we also have a two-dimensional matrix representing DRCS $\mathbf{S}^{(k)}$, which is denoted by 
\begin{align}
  \mathbf{S}^{(k)} = 
  \begin{bmatrix}
    \mathbf{s}_0^{(k)} \\
    \mathbf{s}_1^{(k)} \\
    \vdots \\
    \mathbf{s}_{M-1}^{(k)} 
  \end{bmatrix} = 
  \begin{bmatrix}
    s_{0,0}^{(k)} & s_{0,1}^{(k)} & \cdots & s_{0,N-1}^{(k)} \\
    s_{1,0}^{(k)} & s_{1,1}^{(k)} & \cdots & s_{1,N-1}^{(k)} \\
    \vdots & \vdots & \ddots & \vdots \\
    s_{M-1,0}^{(k)} & s_{M-1,1}^{(k)} & \cdots & s_{M-1,N-1}^{(k)}
  \end{bmatrix}.
\end{align}

For two DRCSs $\mathbf{S}^{(u)}$ and $\mathbf{S}^{(v)}$, their PCAF is defined as the sums of the PCAFs of all subsequences, i.e., 
\begin{align}
  AF_{\mathbf{S}^{(u)},\mathbf{S}^{(v)}}(\tau,f) = 
  \sum_{m=0}^{M-1}AF_{\mathbf{s}_{m}^{(u)},\mathbf{s}_{m}^{(v)}}(\tau,f),
\end{align}
it is abbreviated to $AF_{\mathbf{S}^{(u)}}(\tau,f)$ when $u=v$. For a DRCSS, given a low ambiguity zone (LAZ) $\Pi\subseteq(-N,N)\times(-N,N)$, the maximum periodic ambiguity magnitude of DRCSS $\mathcal{S}$ over this region $\Pi$ is defined as $\theta_{\max}=\{\theta_a,\theta_c\}$, where 
\begin{align}
  \theta_a = \max\{|AF_{\mathbf{C}}(\tau,f)|:\mathbf{C}\in\mathcal{S},(\tau,f)\ne(0,0)\in\Pi\}
\end{align}
is maximal periodic auto-ambiguity magnitude and 
\begin{align}
  \theta_c = \max\{|AF_{\mathbf{C},\mathbf{D}}(\tau,f)|:\mathbf{C},\mathbf{D}\in\mathcal{S}, (\tau,f)\in\Pi\}
\end{align}
is maximal periodic cross-ambiguity magnitude. 

\begin{Remark}
  While the formal concept of DRCSs was recently introduced in \cite{shen2025drcs}, its foundational principle traces back to Turyn's seminal 1963 work on AF properties of complementary sequences. In \cite{turyn1963ambi}, Turyn established the mean square value of the AF for binary sequence pairs--a precursor to DRCS frameworks.
\end{Remark}

A DRCSS with maximum ambiguity magnitude $\theta_{\max}$ over LAZ $\Pi$ is denoted by $(K,M,N,\theta_{\max},\Pi)$-LAZ-DRCSS. For a LAZ-DRCSS, if $\Pi=(-N,N)\times(-N,N)$, it is recorded as $(K,M,N,\theta_{\max})$-DRCSS; if $\theta_{\max}=0$, the LAZ is called the zero ambiguity zone (ZAZ), and it is denoted by $(K,M,N,\Pi)$-ZAZ-DRCSS. Especially, when $M=1$, the DRCSS degenerates into a DRSS \cite{ye2022low} and denoted by $(K,N,\theta_{\max},\Pi)$-DRSS.

\begin{Lemma}[\cite{shen2025drcs}]
  For a $(K,M,N,\theta_{\max},\Pi)$-LAZ-DRCSS, where $\Pi=(-Z_x,Z_x)\times(-Z_y,Z_y)$ and $1\le Z_x,Z_y\le N$, then we have
  \begin{align}
    \theta_{\max} = \frac{MN}{\sqrt{Z_y}} \sqrt{\frac{\frac{KZ_xZ_y}{MN}-1}{KZ_x-1}}.
  \end{align}
\end{Lemma}

\begin{Remark}
Typically, an optimality factor $\rho$ is used to describe the closeness between the maximum AF magnitude of the LAZ-DRCSS and theoretical lower bound, which is defined by 
\begin{align}
  \rho = \frac{\theta_{\max}}{\frac{MN}{\sqrt{Z_y}} \sqrt{\frac{\frac{KZ_xZ_y}{MN}-1}{KZ_x-1}}}.
\end{align}
Obviously, $\rho\ge1$. If $\rho=1$, the LAZ-DRCSS is said to be optimal. If $1<\rho<2$, the LAZ-DRCSS is said to be near-optimal. Specifically, when the parameter approaches infinity and the optimality factor is equal to 1, we refer to it as asymptotically optimal.
\end{Remark}

\begin{Lemma}[\cite{shen2025drcs}]
  For a $(K,M,N,\Pi)$-ZAZ-DRCSS, where $\Pi=(-Z_x,Z_x)\times(-Z_y,Z_y)$ and $1\le Z_x,Z_y\le N$, then we have
  \begin{align}\label{eq-zazdrcss-bound}
    KZ_xZ_y\le MN.
  \end{align}
\end{Lemma}

\begin{Remark}
  For a $(K,M,N,\Pi)$-ZAZ-DRCSS, if $KZ_xZ_y=MN$, the ZAZ-DRCSS is said to be optimal.
\end{Remark}

\subsection{One-Coincidence Frequency-Hopping Sequence Sets}
Let $\mathbb{F}=\{f_0,f_1,\cdots,f_{Q-1}\}$, where $f_i$ is a frequency point, and $Q$ is a positive integer. Usually, $\mathbb{F}$ is equivalent to $\mathbb{Z}_Q$.
Let $\mathbf{X}=\{\mathbf{x}_{k}:0\le k<K\}$ be a frequency hopping sequence set (FHSS) of length $N$ over $\mathbb{F}$, where $\mathbf{x}_k = (x_{k,0},x_{k,1},\cdots,x_{k,N-1})$. The Hamming cross-correlation function of $\mathbf{x}\in\mathbf{X}$ and $\mathbf{y}\in\mathbf{X}$ at shift $\tau$ is defined by
\begin{align}
  H_{\mathbf{x},\mathbf{y}}(\tau) = \sum_{i=0}^{N-1} h(x_{i},y_{i+\tau}),
\end{align}
where $0\le\tau\le N-1$, the summation $i+\tau$ is calculated on $\mathbb{Z}_N$, and $h(a,b)=0$ if $a\ne b$ and $1$ otherwise. If $\mathbf{x}=\mathbf{y}$, then $H_{\mathbf{x},\mathbf{y}}(\tau)$ is called the Hamming auto-correlation function. Further, its maximum Hamming correlation is defined by $H_{\max}=\{H_a,H_c\}$, where  
\begin{align}
  H_a =& \max\{H_{\mathbf{x}}(\tau):\mathbf{x}\in\mathbf{X},0<\tau<N\}, \\
  H_c =& \max\{H_{\mathbf{x},\mathbf{y}}(\tau):\mathbf{x},\mathbf{y}\in\mathbf{X},\forall\tau\},
\end{align}
then $\mathbf{X}$ is referred to as $(K,N,H_{\max},Q)$-FHSS. A FHSS is called one-coincidence FHSS (OC-FHSS) if $H_a=0$ and $H_c=1$, and denoted by $(K,N,Q)$-OC-FHSS.

\subsection{Circular Florentine Rectangle}
\begin{Definition}
  A Tuscan-$k$ rectangle of size $r\times N$ have $r$ rows and $N$ columns such that
  \begin{itemize}
    \item[C1] each row is a permutation of the $N$ symbols and
    \item[C2] for any two distinct symbols $a$ and $b$ and for each $1\le m\le k$, there is at most one row in which $b$ is $m$ steps to the right of $a$,
  \end{itemize}
  it is called a Florentine rectangle (FR) when $k=N-1$. If the circularly-shifted versions of the rows of a Florentine rectangle satisfying the condition that $b$ is $N-m$ steps to the right of $a$ is equivalent to the fact that $b$ is $m$ steps to the left of $a$, it is called a circular FR (CFR).
\end{Definition}

For any positive integer $N$, let $C(N)$ denote the largest integer such that CFR of size $C(N)\times N$ exists. Some known results about $C(N)$ are as follows \cite{song1991on}:
\begin{Lemma}
  For $N\ge2$, we have the following bounds for $C(N)$:
  \begin{itemize}
    \item $C(N)=1$ when $N$ is even,
    \item $p-1\le C(N)\le N-1$, where $p$ is the smallest prime factor of $N$,
    \item $C(N)=N-1$ when $N$ is a prime,
    \item $C(N)\le N-3$ when $N=15\pmod{18}$.
  \end{itemize}
\end{Lemma}

%
\begin{Lemma}\label{lem-cfr}
  Let $N$ be a positive integer, and $p$ is the smallest prime factor of $N$. Define a rectangle $\mathbf{F}=[f_{i,j}]$ of size $(p-1)\times N$, $f_{i,j}=(i+1)\times j \pmod N$, $0\le i<p-1$, $0\le j<N$, then $\mathbf{F}$ is a CFR.
\end{Lemma}

\subsection{Almost Difference Sets}
%
%

\begin{Definition}
  Let $\mathbb{G}$ be an additively written group of order $N$ and $\mathbb{D}$ be a $M$-subset of $\mathbb{G}$. If $t$ nonzero elements of $\mathbb{G}$ has exactly $\lambda$ representations as a difference $d-d'$ with elements from $\mathbb{D}$, and the remaining $N-1-t$ nonzero elements has exactly $\lambda+1$ representations, then $\mathbb{D}$ is said to be a $(N,M,\lambda,t)$-almost difference set (ADS). 
\end{Definition}

Obviously, we have $M(M-1)=t\lambda+(N-1-t)(\lambda+1)$. Moreover, we have the following lemma.
\begin{Lemma}\label{lem-ads-sum}
   For a $(N,M,\lambda,t)$-ADS $\mathbb{D}$ in $\mathbb{Z}_N$ and $\tau\ne0\pmod N$, we have
  \begin{align}
    \left| \sum_{d\in \mathbb{D}}\xi_N^{\tau d} \right| < \sqrt{N+M-\lambda-t-1}.
  \end{align}
\end{Lemma}

For ADSs, the estimate given by Lemma \ref{lem-ads-sum} is too rough. Below we introduce a special ADS based on cyclotomic class, which is similar to DS and the corresponding partial exponential sum can be calculated exactly.

Let $q=p^n=ef+1$ be a power of prime, and let $\alpha$ be a primitive element of $\text{GF}(q)$, where $\text{GF}(q)$ is Galois field with $q$ elements. The cyclotomic class of order $e$ is defined as $\mathbb{D}_i^{(e,q)}=\{\alpha^{me+i}\}$, $0\le m<f$. Based on cyclotomic class, Arasu \emph{et al.} gave a large number of ADSs in \cite{arasu2001almost}, one of which is given below.
\begin{Lemma}\label{lem-con-ads}
  If $q=1\pmod 4$ is a prime power, then $\mathbb{D}_0^{(2,q)}$ is a $(q,\frac{q-1}{2},\frac{q-5}{4},\frac{q-1}{2})$-ADS in $\text{GF}(q)$. Let $N=q=p$ and $\mathbb{D}=\mathbb{D}_0^{(2,q)}$, we have 
  \begin{align}
    \left| \sum_{d\in \mathbb{D}}\xi_N^{\tau d} \right| \le \frac{\sqrt{N}+1}{2}.
  \end{align}
\end{Lemma}

\section{Proposed Constructions}
In this section, we will propose several constructions of DRCSSs based on the ADS and some known sequences, which can generate asymptotically optimal or near-optimal DRCSSs.
\subsection{DRCSSs From One-Coincidence Frequency-Hopping Sequence Sets}
This subsection will propose a method to construct DRCSSs using OC-FHSSs. Based on the known OC-FHSSs, we can obtain many asymptotically optimal DRCSSs with new parameters.
\begin{Theorem}\label{the-drcs-ocfhss}
  Let $\mathbf{F}=\{\mathbf{f}_k:0\le k<K\}$ be a $(K,N,Q)$-OC-FHSS, where $N$, $K$, $Q$ are three positive integers. For each $0\le k<K$, define a DRCS $\mathbf{S}^{(k)}=\{\mathbf{s}_m^{(k)}:0\le m<Q\}$, where $\mathbf{s}_m^{(s)}=(s_{m,0}^{(k)},s_{m,1}^{(k)},\cdots,s_{m,N-1}^{(k)})$, and 
  \begin{align}
    s_{m,n}^{(k)} = \xi_Q^{m\times f_{k,n}},
  \end{align}
  where $0\le m<Q$, $0\le n<N$, $f_{k,n}$ is $n$-th element of $\mathbf{f}_k$. Then $\mathcal{S}=\{\mathbf{S}^{(k)}:0\le k<K\}$ is a periodic $(K,N,N,Q)$-DRCSS.
\end{Theorem}
\begin{proof}
  For any two DRCSs $\mathbf{S}^{(u)}$ and $\mathbf{S}^{(v)}$ in $\mathcal{S}$, where $0\le u,v<K$, their AF is calculated as
  \begin{align}
    AF_{\mathbf{S}^{(u)},\mathbf{S}^{(v)}}(\tau,f) 
    &= \sum_{m=0}^{Q-1}\sum_{i=0}^{N-1}\xi_Q^{m(f_{u,i}-f_{v,i+\tau})}\xi_N^{if} \nonumber \\
    &= \sum_{i=0}^{N-1}\xi_N^{if}\sum_{m=0}^{Q-1}\xi_Q^{m(f_{u,i}-f_{v,i+\tau})}.
  \end{align}
  
  For $\tau=0$ and $u=v$, $AF_{\mathbf{S}^{(u)},\mathbf{S}^{(v)}}(\tau,f)=0$ is trivial when $f\ne0$. For $\tau\ne0$ or $\tau=0$, $u\ne v$, since $\mathbf{F}$ is an OC-FHSS, hence $f_{u,i}-f_{v,i+\tau}=0$ has at most one solution.
  If there is no solution, then $AF_{\mathbf{S}^{(u)},\mathbf{S}^{(v)}}(\tau,f)=0$, which corresponds to the case $u=v$ and $(\tau,f)\ne(0,0)$. If there is one solution $i'$, then we have $AF_{\mathbf{S}^{(u)},\mathbf{S}^{(v)}}(\tau,f) = \xi_N^{i'f} Q$. This completes the proof.
\end{proof}

\begin{Remark}
  Obviously, Theorem \ref{the-drcs-ocfhss} depends on the existence of OC-FHSSs. Some known  parameters of OC-FHSSs are given in Table \ref{tab-known-ocfhss}, where $p$ is a prime, $m$ is a positive number, $U$ is the size of a difference unit set, $e|(p^m-1)$. In \cite[Theorem 7]{shen2025drcs}, we constructed DRCSSs using CFRs. Since CFR is a special OC-FHSS, \cite[Theorem 7]{shen2025drcs} can be regarded as a special case of Theorem \ref{the-drcs-ocfhss}.
\begin{table}[!h]
\centering
\caption{Known OC-FHSSs.}\label{tab-known-ocfhss}
\begin{tabular}{|c|c|c|c|}
\hline
Source & Size $K$ & Length $N$ & Alphabet size $Q$ \\ \hline
\cite{shaar1984survey} & $p$ & $p-1$ & $p$  \\ \hline
\cite{shaar1984survey} & $2^{m}-2$ & $2^{m}-1$ & $2^m-1$  \\ \hline
\cite{titlebaum1981time} & $p-1$ & $p$ & $p$  \\ \hline
\cite{reed1971kth} & $p^m$ & $p^m-1$ & $p^m$  \\ \hline
\cite{cao2006combinatorial}  & $U-1$ & $U$ & $U$  \\ \hline
\cite{cao2006combinatorial} & $e$ & $(p^m-1)/e$ & $p^m$ \\ \hline
\cite{lee2018new} & $p$ & $p(p-1)$ & $p^2$ \\ \hline
\end{tabular}
\end{table}
\end{Remark}

\begin{Corollary}\label{cor-2}
  From Table \ref{tab-known-ocfhss}, the following representative DRCSSs can be obtained:
  \begin{itemize}
    \item Case 1: Based on $(2^m-2,2^m-1,2^m-1)$-OC-FHSS and Theorem \ref{the-drcs-ocfhss}, the DRCSS with parameter $(2^m-2,2^m-1,2^m-1,2^m-1)$ is obtained, which is asymptotically optimal.
    \item Case 2: Based on $(p^m,p^m-1,p^m)$-OC-FHSS and Theorem \ref{the-drcs-ocfhss}, the DRCSS with parameter $(p^m,p^m,p^m-1,p^m)$ is obtained, which is asymptotically optimal.
    \item Case 3: Based on $(p,p(p-1),p^2)$-OC-FHSS and Theorem \ref{the-drcs-ocfhss}, the DRCSS with parameter $(p,p^2,p(p-1),p^2)$ is obtained, which is asymptotically optimal.
  \end{itemize}
\end{Corollary}
\begin{proof}
  The parameters of Case 1 and Case 2 are similar, and we only prove Case 2. For Case 2, the optimality factor for the generated DRCSS is 
  \begin{align}
    \rho_{\text{C2}} 
    &= \frac{\theta_{\max}}{M\sqrt{N} \sqrt{\frac{KN/M-1}{KN-1}} } \nonumber \\
    &= \frac{p^m}{ p^m\sqrt{p^m-1} \sqrt{ \frac{p^m(p^m-1)/p^m-1}{p^m(p^m-1)-1} } } \nonumber \\
    &= \sqrt{\frac{p^m(p^m-1)-1}{(p^m-2)(p^m-1)}}\nonumber \\
    &= 1,~p\rightarrow\infty,
  \end{align}
  which means that the DRCSS obtained by Case 2 is asymptotically optimal. For Case 3, the optimality factor for the generated DRCSS is 
  \begin{align}
    \rho_{\text{C3}} 
    &= \frac{\theta_{\max}}{M\sqrt{N} \sqrt{\frac{KN/M-1}{KN-1}} } \nonumber \\
    &= \frac{p^2}{ p^2\sqrt{p(p-1)} \sqrt{ \frac{p(p(p-1))/p^2-1}{p(p(p-1))-1} } } \nonumber \\
    &= \sqrt{\frac{p^2(p-1)-1}{p(p-1)(p-2)}}\nonumber \\
    &= 1,~p\rightarrow\infty,
  \end{align}
  which means that the DRCSS obtained by Case 3 is asymptotically optimal.
\end{proof}

\begin{Example}\label{exa1-ocfhss}
  Taking $p=5$ and $m=2$, we can get a $(25,24,25)$-OC-FHSS from \cite{reed1971kth}. Based on Theorem \ref{the-drcs-ocfhss}, a DRCSS $\mathcal{S}=\{\mathbf{S}^{(k)}:0\le k<K\}$ with parameter $(25,25,24,25)$ is obtained, and the auto-AF and cross-AF of $\mathbf{S}^{(11)}$ and $\mathbf{S}^{(22)}$ are shown in Fig. \ref{fig-exa1-ocfhss}. Since the maximum Hamming auto-correlation of OC-FHSS is $0$, we can see from Fig. \ref{fig-exa1-ocfhss} that the auto-AF of DRCSs in $\mathcal{S}$ is an ideal thumbtack type, which is not achievable with classical DRS.
\end{Example}

\begin{figure*}[!t]
    \centering
    \subfigure[Auto-AF of $\mathbf{S}^{(11)}$]{
    \begin{minipage}[t]{0.3\linewidth}
    \centering
    \includegraphics[width=5.5cm]{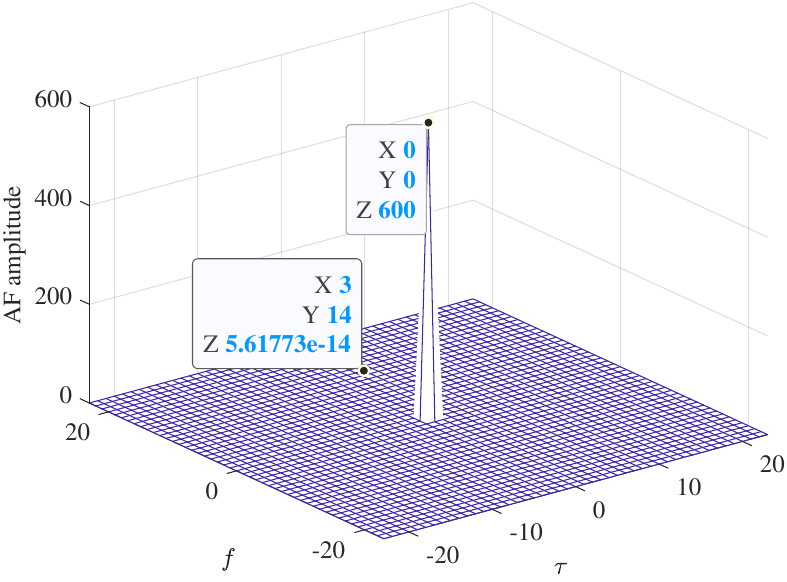}
    \end{minipage}
    }
    \quad
    \subfigure[Auto-AF of $\mathbf{S}^{(22)}$]{
    \begin{minipage}[t]{0.3\linewidth}
    \centering
    \includegraphics[width=5.5cm]{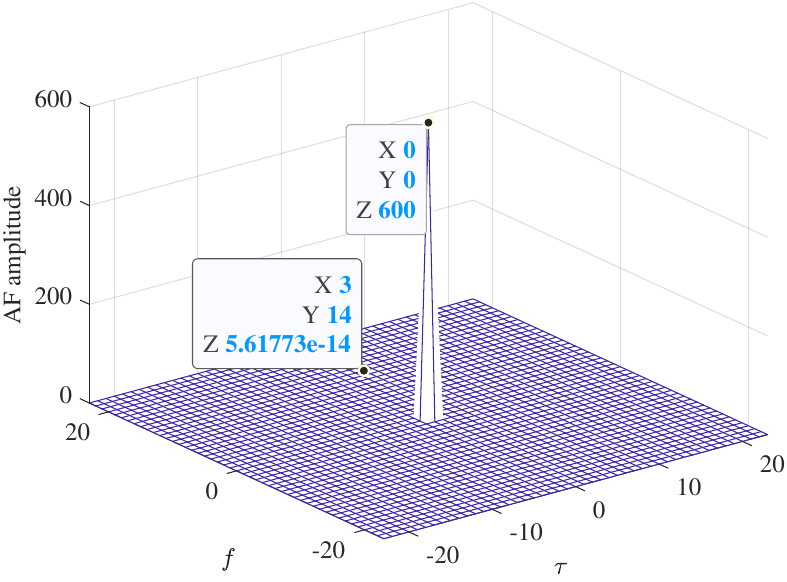}
    \end{minipage}
    }
    \quad
    \subfigure[Cross-AF of $\mathbf{S}^{(11)}$ and $\mathbf{S}^{(22)}$]{
    \begin{minipage}[t]{0.3\linewidth}
    \centering
    \includegraphics[width=5.5cm]{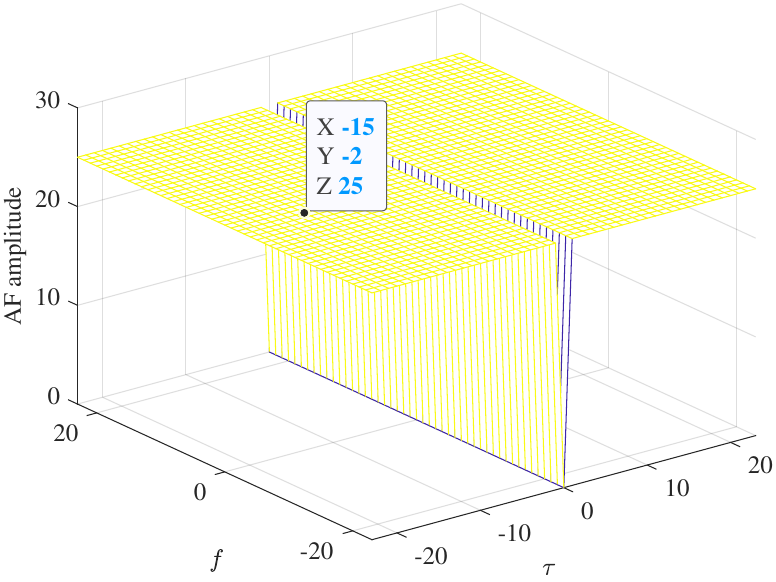}
    \end{minipage}
    }
    \caption{Periodic AFs of $\mathbf{S}^{(11)}$ and $\mathbf{S}^{(22)}$ in Example \ref{exa1-ocfhss}.}\label{fig-exa1-ocfhss}
\end{figure*}

\begin{Theorem}\label{the-cfr}
  For any positive integer $N$, let $K=p_0-1$ and $\mathbf{F}$ be a CFR obtained from Lemma \ref{lem-cfr}, where $p_0$ is the smallest prime factor of $N$.
  For $0\le k<K$, $\mathbf{f}_k$ is $k$-th row of $\mathbf{F}$, define a DRCS $\mathbf{S}^{(k)}=\{\mathbf{s}_m^{(k)}:m\in\mathbb{Z}_N^{(t)}\}$, where $\mathbf{s}_m^{(k)}=(s_{m,0}^{(k)},s_{m,1}^{(k)},\cdots,s_{m,N-1}^{(k)})$, $s_{m,n}^{(k)}=\xi_N^{m\times f_{k,n}}$, $\mathbb{Z}_N^{(t)}=\mathbb{Z}_N/\{t\}$ and $t\in\mathbb{Z}_N$. 
  Then $\mathcal{S}=\{\mathbf{S}^{(k)}:0\le k<K\}$ is a asymptotically optimal $(K,N-1,N,N)$-DRCSS.
\end{Theorem}
\begin{proof}
  For any two DRCSs $\mathbf{S}^{(u)}$ and $\mathbf{S}^{(v)}$ in $\mathcal{S}$, their AF is calculated as
  \begin{align}
    AF_{\mathbf{S}^{(u)},\mathbf{S}^{(v)}}(\tau,f) 
    &= \sum_{m\in\mathbb{Z}_N^{(t)}}\sum_{i\in\mathbb{Z}_N}\xi_N^{m(f_{u,i}-f_{v,i+\tau})}\xi_N^{if} \nonumber \\
    &= \sum_{i\in\mathbb{Z}_N}\xi_N^{if}\sum_{m\in\mathbb{Z}_N^{(t)}}\xi_N^{m(f_{u,i}-f_{v,i+\tau})} \nonumber \\
    &= \sum_{i\in\mathbb{Z}_N}\xi_N^{if}\sum_{m\in\mathbb{Z}_N^{(t)}}\xi^{m((u-v)i-\tau(v+1))}.
  \end{align}
  
  Case 1: $\tau=0$. Trivially, $AF_{\mathbf{S}^{(u)},\mathbf{S}^{(v)}}(\tau,f)=0$ when $u=v$ and $f\ne0$. For $u\ne v$, since $(u-v)i\ne0\pmod N$, so we have
  \begin{align}
    |AF_{\mathbf{S}^{(u)},\mathbf{S}^{(v)}}(\tau,f)| = \left|\sum_{i\in\mathbb{Z}_N} \xi_N^{if} \xi_N^{t(u-v)i}\right| \le N.
  \end{align}
  
  Case 2: $\tau\ne0$. For $u=v$, since $\tau(v+1)\ne0\pmod N$, so we have $|AF_{\mathbf{S}^{(u)},\mathbf{S}^{(v)}}(\tau,f)| = |\xi_N^{t\tau(v+1)}\sum_{i\in\mathbb{Z}_N} \xi_N^{if}|\le N$. For $u\ne v$, we have
  \begin{align}
    AF_{\mathbf{S}^{(u)},\mathbf{S}^{(v)}}(\tau,f)
    = \sum_{m\in\mathbb{Z}_N^{(t)}}\xi_N^{-m\tau(v+1)} \sum_{i\in\mathbb{Z}_N}\xi_N^{(f+m(u-v))i}.
  \end{align}
  Since $0\le u\ne v<p_0-1$, then $(u-v)^{-1}\ne0\pmod N$, which indicates that only when $m=-f(u-v)^{-1}\pmod N$, $\sum_{i\in\mathbb{Z}_N}\xi_N^{(f+m(u-v))i}=N$; otherwise, $\sum_{i\in\mathbb{Z}_N}\xi_N^{(f+m(u-v))i}=0$. Thus, $|AF_{\mathbf{S}^{(u)},\mathbf{S}^{(v)}}(\tau,f)|\le N$.
  
  On the other hand, the optimality factor for the DRCSS generated by Theorem \ref{the-cfr} is 
  \begin{align}
    \rho_{\text{CFR}} 
    &= \frac{\theta_{\max}}{M\sqrt{N} \sqrt{\frac{KN/M-1}{KN-1}} } \nonumber \\
    &= \frac{N}{ (N-1)\sqrt{N} \sqrt{ \frac{p_0N/(N-1)-1}{p_0N-1} } } \nonumber \\
    &= 1,~p_0\rightarrow\infty,
  \end{align}
  so $\mathcal{S}$ is asymptotically optimal.
\end{proof}

\begin{Example}\label{exa-2}
  Let $N=11$, $\mathbf{F}$ be a CFR of size $10\times11$ from Lemma \ref{lem-cfr} and $t=0$. Based on Theorem \ref{the-cfr}, we get a DRCSS $\mathcal{S}$ as shown in Table \ref{tab-exa-cfr}, where each element represents a power of $\xi_{11}$ and ``$\bar{0}$'' is recorded as 10. The auto-AF and cross-AF of $\mathbf{S}^{(2)}$ and $\mathbf{S}^{(7)}$ are shown in Fig. \ref{fig-exa-cfr}.
\end{Example}
\begin{table*}[!t]
\centering
\caption{The DRCSS $\mathcal{S}$ in Example \ref{exa-2}.}\label{tab-exa-cfr}
\begin{tabular}{|c|c|c|c|c|}
\hline
$\mathbf{S}^{(0)}$ & $\mathbf{S}^{(1)}$ & $\mathbf{S}^{(2)}$ & $\mathbf{S}^{(3)}$ & $\mathbf{S}^{(4)}$ \\ \hline
$\begin{bmatrix}
     0     1     2     3     4     5     6     7     8     9    \bar{0}\\
     0     2     4     6     8    \bar{0}     1     3     5     7     9\\
     0     3     6     9     1     4     7    \bar{0}     2     5     8\\
     0     4     8     1     5     9     2     6    \bar{0}     3     7\\
     0     5    \bar{0}     4     9     3     8     2     7     1     6\\
     0     6     1     7     2     8     3     9     4    \bar{0}     5\\
     0     7     3    \bar{0}     6     2     9     5     1     8     4\\
     0     8     5     2    \bar{0}     7     4     1     9     6     3\\
     0     9     7     5     3     1    \bar{0}     8     6     4     2\\
     0    \bar{0}     9     8     7     6     5     4     3     2     1
\end{bmatrix}$ &
$\begin{bmatrix}
     0     2     4     6     8    \bar{0}     1     3     5     7     9\\
     0     4     8     1     5     9     2     6    \bar{0}     3     7\\
     0     6     1     7     2     8     3     9     4    \bar{0}     5\\
     0     8     5     2    \bar{0}     7     4     1     9     6     3\\
     0    \bar{0}     9     8     7     6     5     4     3     2     1\\
     0     1     2     3     4     5     6     7     8     9    \bar{0}\\
     0     3     6     9     1     4     7    \bar{0}     2     5     8\\
     0     5    \bar{0}     4     9     3     8     2     7     1     6\\
     0     7     3    \bar{0}     6     2     9     5     1     8     4\\
     0     9     7     5     3     1    \bar{0}     8     6     4     2
\end{bmatrix}$ &
$\begin{bmatrix}
     0     3     6     9     1     4     7    \bar{0}     2     5     8\\
     0     6     1     7     2     8     3     9     4    \bar{0}     5\\
     0     9     7     5     3     1    \bar{0}     8     6     4     2\\
     0     1     2     3     4     5     6     7     8     9    \bar{0}\\
     0     4     8     1     5     9     2     6    \bar{0}     3     7\\
     0     7     3    \bar{0}     6     2     9     5     1     8     4\\
     0    \bar{0}     9     8     7     6     5     4     3     2     1\\
     0     2     4     6     8    \bar{0}     1     3     5     7     9\\
     0     5    \bar{0}     4     9     3     8     2     7     1     6\\
     0     8     5     2    \bar{0}     7     4     1     9     6     3
\end{bmatrix}$ &
$\begin{bmatrix}
     0     4     8     1     5     9     2     6    \bar{0}     3     7\\
     0     8     5     2    \bar{0}     7     4     1     9     6     3\\
     0     1     2     3     4     5     6     7     8     9    \bar{0}\\
     0     5    \bar{0}     4     9     3     8     2     7     1     6\\
     0     9     7     5     3     1    \bar{0}     8     6     4     2\\
     0     2     4     6     8    \bar{0}     1     3     5     7     9\\
     0     6     1     7     2     8     3     9     4    \bar{0}     5\\
     0    \bar{0}     9     8     7     6     5     4     3     2     1\\
     0     3     6     9     1     4     7    \bar{0}     2     5     8\\
     0     7     3    \bar{0}     6     2     9     5     1     8     4
\end{bmatrix}$ &
$\begin{bmatrix}
     0     5    \bar{0}     4     9     3     8     2     7     1     6\\
     0    \bar{0}     9     8     7     6     5     4     3     2     1\\
     0     4     8     1     5     9     2     6    \bar{0}     3     7\\
     0     9     7     5     3     1    \bar{0}     8     6     4     2\\
     0     3     6     9     1     4     7    \bar{0}     2     5     8\\
     0     8     5     2    \bar{0}     7     4     1     9     6     3\\
     0     2     4     6     8    \bar{0}     1     3     5     7     9\\
     0     7     3    \bar{0}     6     2     9     5     1     8     4\\
     0     1     2     3     4     5     6     7     8     9    \bar{0}\\
     0     6     1     7     2     8     3     9     4    \bar{0}     5
\end{bmatrix}$ \\ \hline
$\mathbf{S}^{(5)}$ & $\mathbf{S}^{(6)}$ & $\mathbf{C}^{(S)}$ & $\mathbf{C}^{(S)}$ & $\mathbf{C}^{(S)}$ \\ \hline
$\begin{bmatrix}
     0     6     1     7     2     8     3     9     4    \bar{0}     5\\
     0     1     2     3     4     5     6     7     8     9    \bar{0}\\
     0     7     3    \bar{0}     6     2     9     5     1     8     4\\
     0     2     4     6     8    \bar{0}     1     3     5     7     9\\
     0     8     5     2    \bar{0}     7     4     1     9     6     3\\
     0     3     6     9     1     4     7    \bar{0}     2     5     8\\
     0     9     7     5     3     1    \bar{0}     8     6     4     2\\
     0     4     8     1     5     9     2     6    \bar{0}     3     7\\
     0    \bar{0}     9     8     7     6     5     4     3     2     1\\
     0     5    \bar{0}     4     9     3     8     2     7     1     6
\end{bmatrix}$ &
$\begin{bmatrix}
     0     7     3    \bar{0}     6     2     9     5     1     8     4\\
     0     3     6     9     1     4     7    \bar{0}     2     5     8\\
     0    \bar{0}     9     8     7     6     5     4     3     2     1\\
     0     6     1     7     2     8     3     9     4    \bar{0}     5\\
     0     2     4     6     8    \bar{0}     1     3     5     7     9\\
     0     9     7     5     3     1    \bar{0}     8     6     4     2\\
     0     5    \bar{0}     4     9     3     8     2     7     1     6\\
     0     1     2     3     4     5     6     7     8     9    \bar{0}\\
     0     8     5     2    \bar{0}     7     4     1     9     6     3\\
     0     4     8     1     5     9     2     6    \bar{0}     3     7
\end{bmatrix}$ &
$\begin{bmatrix}
     0     8     5     2    \bar{0}     7     4     1     9     6     3\\
     0     5    \bar{0}     4     9     3     8     2     7     1     6\\
     0     2     4     6     8    \bar{0}     1     3     5     7     9\\
     0    \bar{0}     9     8     7     6     5     4     3     2     1\\
     0     7     3    \bar{0}     6     2     9     5     1     8     4\\
     0     4     8     1     5     9     2     6    \bar{0}     3     7\\
     0     1     2     3     4     5     6     7     8     9    \bar{0}\\
     0     9     7     5     3     1    \bar{0}     8     6     4     2\\
     0     6     1     7     2     8     3     9     4    \bar{0}     5\\
     0     3     6     9     1     4     7    \bar{0}     2     5     8
\end{bmatrix}$ &
$\begin{bmatrix}
     0     9     7     5     3     1    \bar{0}     8     6     4     2\\
     0     7     3    \bar{0}     6     2     9     5     1     8     4\\
     0     5    \bar{0}     4     9     3     8     2     7     1     6\\
     0     3     6     9     1     4     7    \bar{0}     2     5     8\\
     0     1     2     3     4     5     6     7     8     9    \bar{0}\\
     0    \bar{0}     9     8     7     6     5     4     3     2     1\\
     0     8     5     2    \bar{0}     7     4     1     9     6     3\\
     0     6     1     7     2     8     3     9     4    \bar{0}     5\\
     0     4     8     1     5     9     2     6    \bar{0}     3     7\\
     0     2     4     6     8    \bar{0}     1     3     5     7     9
\end{bmatrix}$ &
$\begin{bmatrix}
     0    \bar{0}     9     8     7     6     5     4     3     2     1\\
     0     9     7     5     3     1    \bar{0}     8     6     4     2\\
     0     8     5     2    \bar{0}     7     4     1     9     6     3\\
     0     7     3    \bar{0}     6     2     9     5     1     8     4\\
     0     6     1     7     2     8     3     9     4    \bar{0}     5\\
     0     5    \bar{0}     4     9     3     8     2     7     1     6\\
     0     4     8     1     5     9     2     6    \bar{0}     3     7\\
     0     3     6     9     1     4     7    \bar{0}     2     5     8\\
     0     2     4     6     8    \bar{0}     1     3     5     7     9\\
     0     1     2     3     4     5     6     7     8     9    \bar{0}
\end{bmatrix}$ \\ \hline
\end{tabular}
\end{table*}
\begin{figure*}[!t]
    \centering
    \subfigure[Auto-AF of $\mathbf{S}^{(2)}$]{
    \begin{minipage}[t]{0.3\linewidth}
    \centering
    \includegraphics[width=5.5cm]{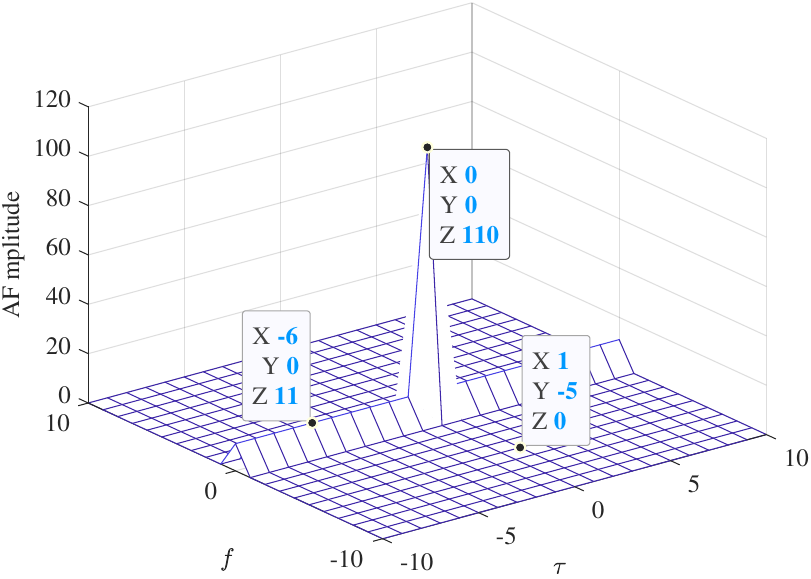}
    \end{minipage}
    }
    \quad
    \subfigure[Auto-AF of $\mathbf{S}^{(7)}$]{
    \begin{minipage}[t]{0.3\linewidth}
    \centering
    \includegraphics[width=5.5cm]{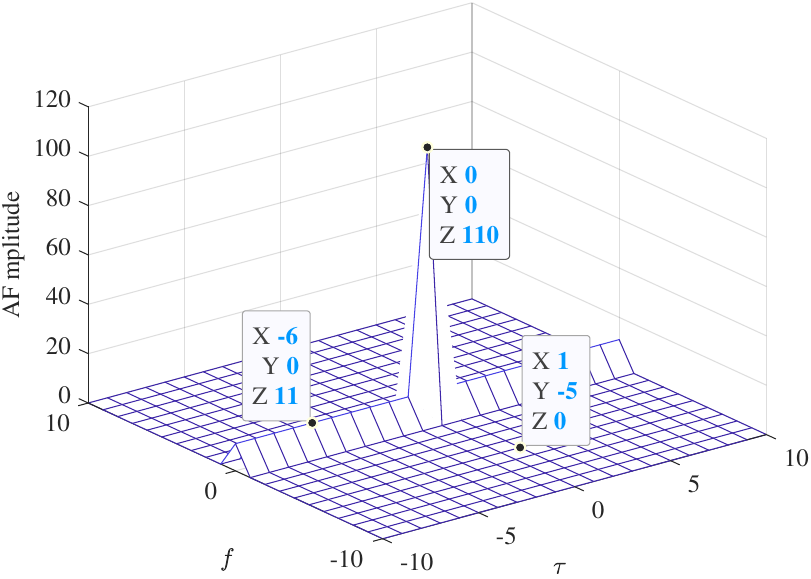}
    \end{minipage}
    }
    \quad
    \subfigure[Cross-AF of $\mathbf{S}^{(2)}$ and $\mathbf{S}^{(7)}$]{
    \begin{minipage}[t]{0.3\linewidth}
    \centering
    \includegraphics[width=5.5cm]{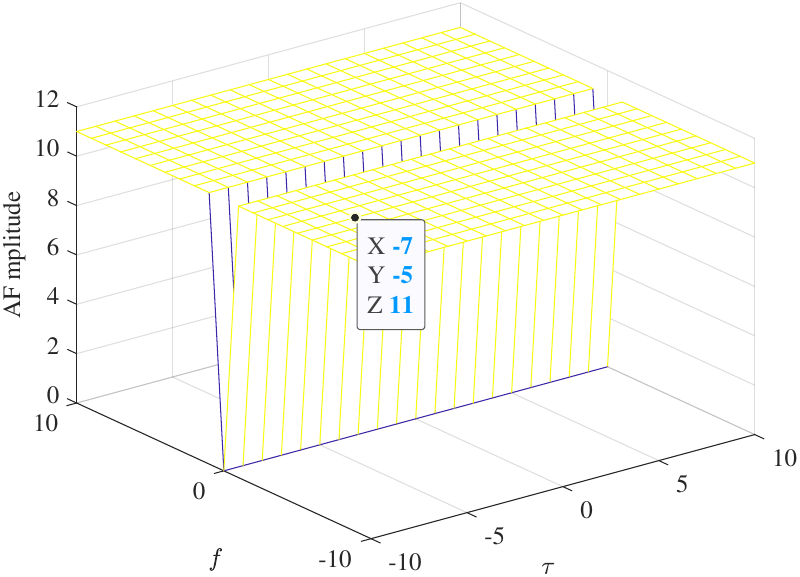}
    \end{minipage}
    }
    \caption{Periodic AFs of $\mathbf{S}^{(2)}$ and $\mathbf{S}^{(7)}$ in Example \ref{exa-2}.}\label{fig-exa-cfr}
\end{figure*}

\subsection{LAZ-DRCSSs From Almost Difference Sets}
In \cite{shen2025drcs}, we proposed a construction of optimal LAZ-DRCSSs based DSs and DRSSs. In this subsection, we will improve this construction so that it can obtain more usable LAZ-DRCSSs. 
\begin{Construction}\label{con-1}
  Let $N$ be a positive integer, $M\le N$, $\mathbb{D}=\{d_m:0\le m<M\}$ be a $M$-subset of $\mathbb{Z}_N$ and $\mathbf{S}=\{\mathbf{s}_k:0\le k<K\}$ be a $(K,N,\alpha_{\max},\Pi)$-DRSS, where $\Pi=(-Z_x,Z_x)\times(-Z_y,Z_y)$. For any $0\le k<K$, define a DRCS $\mathbf{S}^{(k)}=\{\mathbf{s}_m^{(k)}:0\le m<M\}$, where $\mathbf{s}_m^{(k)}=(s_{m,0}^{(k)},s_{m,1}^{(k)},\cdots,s_{m,N-1}^{(k)})$ and $s_{m,n}^{(k)} = s_{k,n} \times \xi_N^{nd_m}$.
\end{Construction}

\begin{Remark}\label{rem-alphamaxfd}
For two DRCS $\mathbf{S}^{(u)}$ and $\mathbf{S}^{(v)}$ with $0\le u,v<K$ in Construction \ref{con-1}, their PCAF is calculated as
\begin{align}
|AF_{\mathbf{S}^{(u)},\mathbf{S}^{(v)}}(\tau,f)|
&=\left|\sum_{m=0}^{M-1}\sum_{n=0}^{N-1}s_{u,n}s_{v,n+\tau}^{*}\xi_N^{-\tau d_m}\xi_N^{nf}\right| \nonumber \\
&=|AF_{\mathbf{s}_u,\mathbf{s}_v}(\tau,f)|\times f(\mathbb{D}),
\end{align}
where $f(\mathbb{D})=|\sum_{m=0}^{M-1}\xi_N^{-\tau d_m}|$. For $\tau\ne0$, we have $|AF_{\mathbf{S}^{(u)},\mathbf{S}^{(v)}}(\tau,f)|=\alpha_{\max}f(\mathbb{D})$. 
For $\tau=0$, when $u=v$, $AF_{\mathbf{s}_u,\mathbf{s}_v}(0,f)=0$ for all $f\ne0$, so we have $|AF_{\mathbf{S}^{(u)},\mathbf{S}^{(v)}}(\tau,f)|=M\times \max_{u\ne v}|AF_{\mathbf{s}_u,\mathbf{s}_v}(0,f)|$. This means that if the DRS set satisfies $\max_{u\ne v}|AF_{\mathbf{s}_u,\mathbf{s}_v}(0,f)|=0$, the maximum AF amplitude of the constructed DRCSS is $\alpha_{\max}f(\mathbb{D})$.
\end{Remark}

\begin{Lemma}\label{the-drs-cubic}
  For a odd prime $N$ and $2\le K\le N$, let $Z_y=\lfloor N/K \rfloor$. Define $\mathbf{S}=\{\mathbf{s}_k:0\le k<K\}$, where $\mathbf{s}_k=(s_{k,0},s_{k,1},\cdots,s_{k,N-1})$ and $s_{k,n} = \xi_N^{n^3+k\lfloor N/K \rfloor n}$ for $0\le n<N$. Then $\mathbf{S}$ is a $(K,N,\sqrt{N},\Pi)$-DRSS, where $\Pi=(-N,N)\times(-Z_y,Z_y)$. Besides, $AF_{\mathbf{s}_u,\mathbf{s}_v}(0,f)=0$ for any $0\le u\ne v<K$ and $f\in(-Z_y,Z_y)$.
\end{Lemma}

\begin{Theorem}\label{the-drcss-ads}
   In Construction \ref{con-1}, if $N=1\pmod 4$ is a prime, $\mathbb{D}$ is a $(N,\frac{N-1}{2},\frac{N-5}{2},\frac{N-1}{2})$-ADS from Lemma \ref{lem-con-ads} and $\mathbf{S}$ is a $(K,N,\sqrt{N},\Pi)$-DRSS from Lemma \ref{the-drs-cubic}, where $\Pi=(-N,N)\times(-Z_y,Z_y)$ and $Z_y=\lfloor N/K \rfloor$. Then, $\mathcal{S}$ is a asymptotically optimal LAZ-DRCSS with parameter $(K,M=\frac{N-1}{2},N,\theta_{\max}=\frac{N+\sqrt{N}}{2},\Pi=(-N,N)\times (-Z_y,Z_y)$.
\end{Theorem}
\begin{proof}
  According to Remark \ref{rem-alphamaxfd}, we have $\theta_{\max}=\alpha_{\max}f(\mathbb{D})=\frac{N+\sqrt{N}}{2}$. The optimality factor for the LAZ-DRCSS generated by Theorem \ref{the-drcss-ads} is
  \begin{align}
    \rho_{\text{ADS}} 
    &= \frac{\theta_{\max}}{\frac{MN}{\sqrt{Z_y}}\sqrt{\frac{\frac{KZ_xZ_y}{MN}-1}{KZ_x-1}}}\nonumber \\
    &\approx \frac{N+\sqrt{N}}{N-1}\nonumber \\
    &=1,~N\rightarrow\infty,
  \end{align}
  so $\mathcal{S}$ is asymptotically optimal.
\end{proof}

In Table \ref{tab-adsyedrs}, we provide some parameters for asymptotically optimal LAZ-DRCSSs when $K=2$ or $K=3$.

\begin{table}[!t]
\centering
\caption{Some parameters of asymptotically optimal LAZ-DRCSS obtained from Theorem \ref{the-drcss-ads}.}\label{tab-adsyedrs}
\begin{tabular}{|c|c|c|c|c|c|c|}
\hline
Length $N$ & Set size $K$ & Flock size $M$ & $\theta_{\max}$ & $Z_x$ & $Z_y$ & Optimality factor \\ \hline
\multirow{2}{*}{29} & 2 & 14 & 17.1926 & 29 & 14 & 1.1962 \\ \cline{2-7} 
                  & 3 & 14 & 17.1926 & 29 & 9 & 1.2226 \\ \hline
\multirow{2}{*}{71} & 2 & 35 & 39.7131 & 71 & 35 & 1.1227 \\ \cline{2-7} 
                  & 3 & 35 & 39.7131 & 71 & 23 & 1.1322 \\ \hline
\multirow{2}{*}{101} & 2 & 50 & 55.5249 & 101 & 50 & 1.1022 \\ \cline{2-7} 
                  & 3 & 50 & 55.5249 & 101 & 33 & 1.1088 \\ \hline
\multirow{2}{*}{149} & 2 & 74 & 80.6033 & 149 & 74 & 1.0837 \\ \cline{2-7} 
                  & 3 & 74 & 80.6033 & 149 & 49 & 1.0881 \\ \hline
\multirow{2}{*}{181} & 2 & 90 & 97.2268 & 181 & 90 & 1.0758 \\ \cline{2-7} 
                  & 3 & 90 & 97.2268 & 181 & 60 & 1.0763 \\ \hline
\multirow{2}{*}{229} & 2 & 114 & 122.0664 & 229 & 114 & 1.0673 \\ \cline{2-7} 
                  & 3 & 114 & 122.0664 & 229 & 76 & 1.0676 \\ \hline
\end{tabular}
\end{table}

The premise of generating LAZ-DRCSS using Construction \ref{con-1} is that there are ADSs and DRS sets corresponding to the parameters, and the restriction of $KZ_xZ_y\ge MN$ must be satisfied, which is not easy to find. In \cite{ding2013unit}, the authors proposed a construction of DRSs of length $p^m-1$ with optimal auto-AF, i.e., $\theta_a=\sqrt{p^m-1}$, where $\Pi=(1-p^m,p^m-1)\times(1-p^m,p^m-1)$. 
Coincidentally, there is almost a ADS on $\mathbb{Z}_{p^m-1}$ with parameter $(p^m-1,\frac{p^m-1}{2},\frac{p^m-3}{4},\frac{3p^m-5}{4})$ for $p^m =3\pmod 4$ and $(p^m-1,\frac{p^m-1}{2},\frac{p^m-5}{4},\frac{p^m-1}{4})$ for $p^m=1\pmod 4$ \cite{lempel1977a}, where $p$ is an odd prime. Therefore, we can draw the following corollary.

\begin{Corollary}\label{cor-ads}
  Let $p$ be an odd prime, $m$ be a positive integer and $q=p^m$. Let $\mathbf{s}$ be a sequence of length $q-1$ with $\theta_a=\sqrt{q-1}$ over $\Pi=(1-q,q-1)\times(1-q,q-1)$, and $\mathbb{D}=\{d_0,d_1,\cdots,d_{M-1}\}$ be a $M$-subset of $\mathbb{Z}_{q-1}$.
  Define a DRCS $\mathbf{S}=\{\mathbf{s}_m:0\le m<M\}$, where $s_{m,n}=s_n\times\xi_N^{n d_m}$ for $0\le n<q-1$. 
  If 
  $\mathbb{D}$ is a $(q-1,\frac{q-1}{2},\frac{q-3}{4},\frac{3q-5}{4})$-ADS for $q=3\pmod 4$ or $(q-1,\frac{q-1}{2},\frac{q-5}{4},\frac{q-1}{4})$-ADS for $q=1\pmod4$, then $\mathbf{S}$ is a near-optimal DRCS with 
  \begin{align}
    \theta_{\max}=\theta_a=
    \begin{cases}
      \frac{q-1}{\sqrt{2}}, & \mbox{if } q=3\pmod4; \\
      q-1, & \mbox{if } q=1\pmod4.
    \end{cases}
  \end{align}
\end{Corollary}

Based on Lemma \ref{lem-ads-sum} and Remark \ref{rem-alphamaxfd}, this corollary is trivial and will not be repeated here. In Table \ref{tab-adsneardrcss}, we provide some parameters for near-optimal DRCSSs.

\begin{table}[!t]
\centering
\caption{Some parameters of near-optimal DRCSS obtained from Corollary \ref{cor-ads} for $q=p-1$.} \label{tab-adsneardrcss}
\begin{tabular}{|c|c|c|c|c|c|c|c|}
\hline
$p$ & Length $N$ & Set size $K$ & Flock size $M$ & $\theta_{\max}$ & $Z_x$ & $Z_y$ & Optimality factor \\ \hline
11& 10 & \multirow{14}{*}{1} & 5 & 7.0711 & 10 & 10 & 1.3416 \\ \cline{1-2} \cline{4-8} 
13&12 &   & 6 & 12 & 12 & 12 & 1.9149 \\ \cline{1-2} \cline{4-8} 
17&16 &   & 8 & 16 & 16 & 16 & 1.9365 \\ \cline{1-2} \cline{4-8} 
19&18 &   & 9 & 12.7279 & 18 & 18 & 1.3744 \\ \cline{1-2} \cline{4-8} 
23&22 &   & 11 & 15.5563 & 22 & 22 & 1.3817 \\ \cline{1-2} \cline{4-8} 
29&28 &   & 14 & 28 & 28 & 28 & 1.9640 \\ \cline{1-2} \cline{4-8} 
31&30 &   & 15 & 21.2132 & 30 & 30 & 1.3904 \\ \cline{1-2} \cline{4-8} 
41&40 &   & 20 & 40 & 40 & 40 & 1.9748 \\ \cline{1-2} \cline{4-8} 
43&42 &   & 21 & 29.6985 & 42 & 42 & 1.3973 \\ \cline{1-2} \cline{4-8} 
47&46 &   & 23 & 32.5269 & 46 & 46 & 1.3988 \\ \cline{1-2} \cline{4-8}
53&52 &   & 26 & 52 & 52 & 52 & 1.9807 \\ \cline{1-2} \cline{4-8} 
59&58 &   & 29 & 41.0122 & 58 & 58 & 1.4020 \\ \cline{1-2} \cline{4-8} 
61&60 &   & 30 & 60 & 60 & 60 & 1.9833 \\ \cline{1-2} \cline{4-8} 
67&66 &   & 33 & 46.6690 & 66 & 66 & 1.4035 \\ \hline
\end{tabular}
\end{table}

\subsection{ZAZ-DRCSSs From Complete Complementary Codes}
In this subsection, we will use complete complementary codes (CCCs) to construct the optimal ZAZ-DRCSS. Below, we will give some brief basic knowledge about CCCs.
\begin{Definition}
  Let $\mathcal{C}=\{\mathbf{C}^{(k)}:0\le k<M\}$ be a sequence set, each $\mathbf{C}^{(k)}$ consists of $M$ subsequences of length $L$, i.e., $\mathbf{C}^{(k)}=\{\mathbf{c}_m^{(k)}:0\le m<M\}$, $\mathbf{c}_m^{(k)}=(c_{m,0}^{(k)},c_{m,1}^{(k)},\cdots,c_{m,N-1}^{(k)})$. For $0\le k,t<M$, if
  \begin{align}
    R_{\mathbf{C}^{(k)},\mathbf{C}^{(t)}}(\tau) 
    = \sum_{m=0}^{M-1} R_{\mathbf{c}_m^{(k)},\mathbf{c}_m^{(t)}}(\tau)=
    \begin{cases}
      MN, & \tau=0, k=t; \\
      0, & \mbox{otherwise},
    \end{cases}
  \end{align}
  then $\mathcal{C}$ is called a CCC and denoted by $(M,N)$-CCC.
\end{Definition}

At present, the research on CCCs is very mature, and some construction methods can be found in \cite{suehiro1988n,chen2008comp,han2011syst,liu2014new,shen2023cons}. It is worth noting that there exists a CCC with parameters of the form $(N,N)$, where $N$ is an arbitrary positive integer.

\begin{Theorem}\label{the-drcssccc}
  Let $\mathcal{C}=\{\mathbf{C}^{(k)}:0\le k<M\}$ be an $(M,N)$-CCC and $\mathbf{a}$ be an all-1 sequence of length $L$. Define a sequence set $\mathcal{S}=\{\mathbf{S}^{(k)}:0\le k<M\}$ and  $\mathbf{S}^{(k)}=\{\mathbf{s}_m^{(k)}:0\le m<M\}$, where
  \begin{align}
    \mathbf{s}_m^{(k)} = \mathbf{a}\otimes\mathbf{c}_m^{(k)}.
  \end{align}
  Then $\mathcal{S}$ be an optimal ZAZ-DRCSS with parameter $(M,M,NL,\Pi)$, where $\Pi=(-Z_x,Z_y)\times(-Z_y,Z_y)$, $Z_x=N$ and $Z_y=L$.
\end{Theorem}

\begin{proof}
  For two DRCS $\mathbf{S}^{(k)}$ and $\mathbf{S}^{(t)}$, we have 
  \begin{align}
    AF_{\mathbf{S}^{(k)},\mathbf{S}^{(t)}}(\tau,f) 
    =& \sum_{m=0}^{M-1}\sum_{i=0}^{NL-1}s_{m,i}^{(k)}(s_{m,i+\tau}^{(t)})^{*}\xi_{NL}^{if} \nonumber \\
    =& \sum_{m=0}^{M-1}\sum_{l=0}^{L-1}\sum_{n=0}^{N-1}
    c_{m,n}^{(k)}(c_{m,n+\tau}^{(t)})^{*}\xi_{NL}^{nf}\xi_{NL}^{lNf} \nonumber\\
    =& \sum_{m=0}^{M-1}\sum_{n=0}^{N-1}c_{m,n}^{(k)}(c_{m,n+\tau}^{(t)})^{*}\xi_{NL}^{nf}
    \sum_{l=0}^{L-1}\xi_{L}^{lf}.
  \end{align}
  Obviously, $AF_{\mathbf{S}^{(k)},\mathbf{S}^{(t)}}(\tau,f)=0$ when $f\ne 0\pmod L$. Therefore, we only need to consider the case where $f=0$. Since $\mathcal{C}$ is a CCC, so we have 
  \begin{align}
    AF_{\mathbf{S}^{(k)},\mathbf{S}^{(t)}}(\tau,0)
    =\begin{cases}
       MNL, & \tau=0~(\text{mod }N),k=t, \\
       0, & \text{otherwise}.
     \end{cases}
  \end{align}
  In summary, $\mathcal{S}$ is a ZAZ-DRCSS, where ZAZ is $\Pi=(-N,N)\times(-L,L)$. On the other hand, it is evident that the parameters of $\mathcal{S}$ can make the equal sign of \eqref{eq-zazdrcss-bound} hold, therefore it is optimal.
\end{proof}

\begin{Example}\label{exa-ccc-drcs}
  Taking a binary $(4,4)$-CCC $\mathcal{C}=\{\mathbf{C}^{(k)}:0\le k\le 3\}$ as an example, it is represented as follows:
  \begin{align}
    \mathbf{C}^{(0)} &= 
    \begin{bmatrix}
      + & - & + & - \\
      + & + & + & + \\
      + & - & - & + \\
      - & - & + & + 
    \end{bmatrix},~
    \mathbf{C}^{(1)} = 
    \begin{bmatrix}
      + & + & + & + \\
      + & - & + & - \\
      + & + & - & - \\
      - & + & + & - 
    \end{bmatrix}, \\
    \mathbf{C}^{(2)} &= 
    \begin{bmatrix}
      + & - & - & + \\
      + & + & - & - \\
      + & - & + & - \\
      - & - & - & - 
    \end{bmatrix},~
    \mathbf{C}^{(3)} = 
    \begin{bmatrix}
      + & + & - & - \\
      + & - & - & + \\
      + & + & + & + \\
      - & + & - & + 
    \end{bmatrix}.
  \end{align}
  Set $L=5$, according to Theorem \ref{the-drcssccc}, we obtain a $(4,4,20,\Pi)$-ZAZ-DRCSS $\mathcal{S}=\{\mathbf{S}^{(k)}:0\le k\le 3\}$ and $\mathbf{S}^{(k)} = [\mathbf{C}^{(k)},\mathbf{C}^{(k)},\mathbf{C}^{(k)},\mathbf{C}^{(k)},\mathbf{C}^{(k)}]$. Taking $\mathbf{S}^{(0)}$ and $\mathbf{S}^{(3)}$ as examples, their PCAF and PAAF are shown in Fig. \ref{fig-exa-ccc}. Clearly, the ZAZ of DRCSS $\mathcal{S}$ is $\Pi=(-4,4)\times(-5,5)$.
\end{Example}

\begin{figure*}[!t]
\centering
\subfigure[Auto-AF of $\mathbf{S}^{(0)}$]{
\begin{minipage}[t]{0.3\linewidth}
\centering
\includegraphics[width=5.5cm]{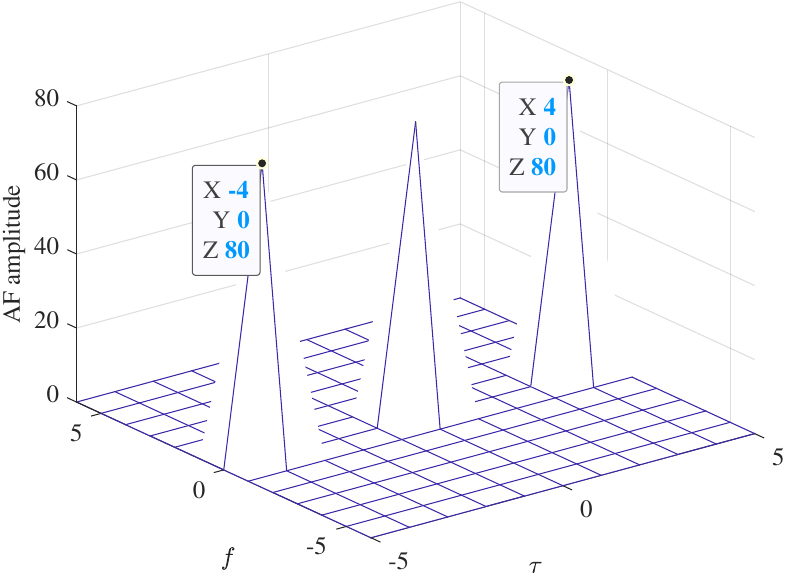}
\end{minipage}
}
\quad
\subfigure[Auto-AF of $\mathbf{S}^{(3)}$]{
\begin{minipage}[t]{0.3\linewidth}
\centering
\includegraphics[width=5.5cm]{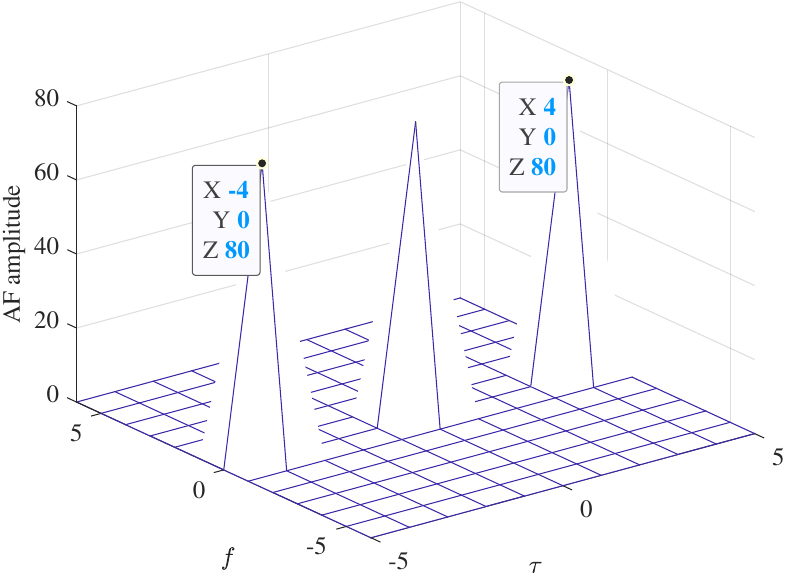}
\end{minipage}
}
\quad
\subfigure[Cross-AF of $\mathbf{S}^{(0)}$ and $\mathbf{S}^{(3)}$]{
\begin{minipage}[t]{0.3\linewidth}
\centering
\includegraphics[width=5.5cm]{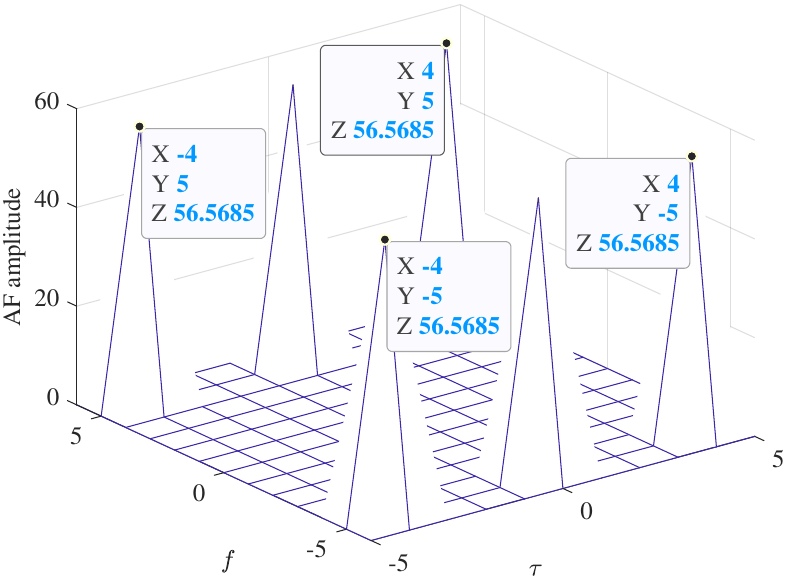}
\end{minipage}
}
\caption{Periodic AFs of $\mathbf{S}^{(0)}$ and $\mathbf{S}^{(3)}$ in Example \ref{exa-ccc-drcs}.}\label{fig-exa-ccc}
\end{figure*}

\section{Conclusion}
This paper is devoted to studying novel unimodular LAZ/ZAZ-DRCSSs with new parameters. First, we proposed a construction of DRCSSs based on OC-FHSSs, and deriving several new DRCSSs using the known OC-FHSS, see Theorem \ref{the-drcs-ocfhss}, Theorem \ref{the-cfr} and Corollary \ref{cor-2}. Then, we proposed the new LAZ-DRCSS by using the (near)-optimal DRSSs and ADSs, see Theorem \ref{the-drcss-ads} and Corollary \ref{cor-ads}. Finally, we generated the ZAZ-DRCSS by repeating the CCCs, see Theorem \ref{the-drcssccc}. These proposed LAZ/ZAZ-DRCSSs are optimal or near-optimal relative to the periodic theoretical bound in \cite{shen2025drcs}. Crucially, the proof of Theorem \ref{the-drcs-ocfhss} reveals that the Hamming auto-correlation requirement for OC-FHSSs can be relaxed from $H_a=0$ to $H_a=1$ while preserving optimality of proposed DRCSSs. Future work will develop weakened OC-FHSSs with larger set sizes to obtain more optimal DRCSSs.

\bibliographystyle{IEEEtran}    
\bibliography{Reference.bib}    
\end{document}